\documentclass[10pt]{article}
\usepackage{fullpage, a4wide}
\usepackage[utf8]{inputenc}
\usepackage{cite}

\usepackage{amssymb}
\usepackage{amsmath}
\usepackage{graphicx}
\usepackage{enumitem}
\usepackage{algorithmic}
\usepackage[linesnumbered, ruled, boxed]{algorithm2e}

\newcommand{\Path}{\ensuremath\mathrm{Path}}
\newcommand{\Prefix}{\ensuremath\mathrm{Prefix}}
\newcommand{\Suffix}{\ensuremath\mathrm{Suffix}}
\newcommand{\Match}[1]{\ensuremath\mathrm{Match}(#1)}

\newcommand{\ceil}[1]{\left\lceil{#1}\right\rceil}

\newcommand{\AO}{\ensuremath{A_O}}
\newcommand{\AI}{\ensuremath{A_I}}

\newcommand{\RT}{\ensuremath{\mathcal{T}_{\textrm{rec}}}}
\newcommand{\PT}{\ensuremath{\mathcal{T}^P}}

\newtheorem{theorem}{Theorem}

\newtheorem{lemma}[theorem]{Lemma}

\newenvironment{proof}{

\noindent{\bf Proof:}} {\hfill$\blacksquare$

}

\title{From Regular Expression Matching to Parsing}

\author{Philip Bille \\ \texttt{phbi@dtu.dk} \and Inge Li G{\o}rtz \\ \texttt{inge@dtu.dk}}

\date{}

\begin{document}
\maketitle

\begin{abstract}
Given a regular expression $R$ and a string $Q$, the regular expression parsing problem is to determine if $Q$ matches $R$ and if so, determine how it matches, e.g., by a mapping of the characters of $Q$ to the characters in $R$. Regular expression parsing makes finding matches of a regular expression even more useful by allowing us to directly extract subpatterns of the match, e.g., for extracting IP-addresses from internet traffic analysis or extracting subparts of genomes from genetic data bases. We present a new general techniques for efficiently converting a large class of algorithms that determine if a string $Q$ matches regular expression $R$ into algorithms that can construct a corresponding mapping. As a consequence, we obtain the first efficient linear space solutions for regular expression parsing. 
\end{abstract}
%\newpage

\section{Introduction}
A regular expression specifies a set of strings formed by characters combined with concatenation, union (\texttt |), and Kleene star (\texttt *) operators. For instance, \texttt{(a|(ba))*)} describes the set of strings of \texttt{a}s and \texttt{b}s, where every \texttt{b} is followed by an \texttt{a}. Regular expressions are a fundamental concept in formal language theory and a basic tool in computer science for specifying search patterns. Regular expression search appears in diverse areas such as internet traffic analysis~\cite{JMR2007, YCDLK2006, KDYCT2006}, data mining~\cite{GRS1999}, data bases~\cite{LM2001, Murata2001}, computational biology~\cite{NR2003}, and human computer interaction~\cite{KHDA2012}.

Given a regular expression $R$ and a string $Q$, the \emph{regular expression parsing problem}~\cite{kearns1991, DF2000,FC2004, NH2011, SM2012, Laurikari2000} is to determine if $Q$ matches a string in the set of strings described by $R$ and if so, determine how it matches by computing the corresponding sequence of positions of characters in $R$, i.e., the mapping of each character in $Q$ to a character in $R$ corresponding to the match. For instance, if $R = $  \texttt{(a|(ba))*)} and $Q = $ \texttt{aaba}, then $Q$ matches $R$ and $1, 1, 2, 3$ is a corresponding parse specifying that $Q[1]$ and $Q[2]$ match the first \texttt{a} in $R$, $Q[3]$ match the \texttt{b} in $R$, and $Q[4]$ match the last \texttt{a} in $R$\footnote{Another typical definition of parsing is to compute a parse tree (or a variant thereof) of the derivation of $Q$ on $R$. Our definition simplifies our presentation and it is straightforward to derive a parse tree from our parses.}. Regular expression parsing makes finding matches of a regular expression even more useful by allowing us to directly extract subpatterns of the match, e.g., for extracting IP-addresses from internet traffic analysis or extracting subparts of genomes from genetic data bases.

To state the existing bounds, let $n$ and $m$ be the length of the string and the regular expression, respectively. As a starting point consider the simpler \emph{regular expression matching problem}, that is, to determine if $Q$ matches a string in the set of strings described by $R$ (without necessarily providing a mapping from characters in $Q$ to characters in $R$). A classic textbook algorithm to matching, due to Thompson~\cite{Thomp1968}, constructs and simulates a non-deterministic finite automaton (NFA) in $O(nm)$ time and $O(m)$ space. An immediate approach to solve the parsing problem is to combine Thompson's algorithm with backtracking. To do so, we store all state-sets produced during the NFA simulation and then process these in reverse order to recover an accepting path in the NFA matching $Q$. From the path we then immediately obtain the corresponding parse of $Q$ since each transition labeled by a character uniquely corresponds to a character in $R$. This algorithm uses $O(nm)$ time and space. Hence, we achieve the same time bound as matching but increase the space by an $\Omega(n)$ factor. We can improve the time by polylogarithmic factors using faster algorithms for matching~\cite{Myers1992, Bille06, BFC2008, BT2009, BT2010}, but by a recent conditional lower bound~\cite{BI2016} we cannot hope to achieve $\Omega((nm)^{1-\varepsilon})$ time assuming the strong exponential time hypothesis. Other direct approaches to regular expression parsing~\cite{kearns1991, DF2000,FC2004, NH2011, SM2012, Laurikari2000} similarly achieve $\Theta(nm)$ time and space (ignoring polylogarithmic factors), leaving a substantial gap between linear space for matching and $\Theta(nm)$ space for parsing. The goal of this paper is to address this gap. 

\subsection{Results}
We present a new technique to efficiently extend the classic state-set transition  algorithms for matching to also solve parsing in the same time complexity while only using linear space. Specifically, we obtain the following main result based on Thompson's algorithm: 
\begin{theorem}\label{thm:main}
	Given a regular expression of length $m$ and a string of length $n$, we can solve the regular expression parsing problem in $O(nm)$ time and $O(n + m)$ space. 
\end{theorem}
This is the first bound to significantly improve upon the combination of $\Theta(nm)$ time and space. The result holds on a comparison-based, pointer machine model of computation. Our techniques are sufficiently general to also handle the more recent faster state-set transition algorithms~\cite{Myers1992, BFC2008, Bille06} and we also obtain a similar space improvement for these.

\subsection{Techniques}
Our overall approach is similar in spirit to the classic divide and conquer algorithm by Hirschberg~\cite{Hirschberg1975} for computing a longest common subsequence of two strings in linear space. Let $A$ be the \emph{Thompson NFA} (TNFA) for $R$ built according to Thompson's rules~\cite{Thomp1968} (see also Fig.~\ref{fig:thompson}) with $m$ states, and let $Q$ be the string of length $n$. 

We first decompose $A$ using standard techniques into a pair of nested subTNFAs called the \emph{inner subTNFA} and the \emph{outer subTNFA}. Each have roughly at most $2/3$ of the states of $A$ and overlap in at most $2$ boundary states. We then show how to carefully simulate $A$ to decompose $Q$ into substrings corresponding to subparts of an accepting path in each of the subTNFAs. The key challenge here is to efficiently handle cyclic dependencies between the subTNFAs. From this we construct a sequence of subproblems for each of the substrings corresponding to the inner subTNFAs and a single subproblem for the outer subTNFA. We recursively solve these to construct a complete accepting path in $A$. This strategy leads to an $O(nm)$ time and $O(n\log m + m)$ space solution. We show how to tune and organize the recursion to avoid storing intermediate substrings leading to the linear space solution in Thm.~\ref{thm:main}. Finally, we show how to extend our solution to obtain linear space parsing solutions for other state-set transition algorithms. 

\section{Preliminaries}

\paragraph*{Strings}
A string Q of length $n = |Q|$ is a sequence $Q[1]\ldots Q[n]$ of $n$ characters drawn from an alphabet $\Sigma$. The string $Q[i]\ldots Q[j]$ denoted $Q[i, j]$ is called a substring of $Q$. The substrings $Q[1, i]$ and $Q[j, n]$ are the $i$th prefix and the $j$th suffix of $Q$, respectively. The string $\epsilon$ is the unique empty string of length zero.

\paragraph*{Regular Expressions}
First we briefly review the classical concepts used in the paper. For
more details see, e.g., Aho et al.~\cite{ASU1986}. We consider the set of non-empty 
regular expressions over an alphabet $\Sigma$, defined recursively as
follows. If $\alpha \in \Sigma \cup \{\epsilon\}$ then $\alpha$ is a regular expression, and if $S$ and $T$ are regular expressions then so is the \emph{concatenation}, $(S)\cdot(T)$, the \emph{union}, $(S)|(T)$, and the \emph{star}, $(S)^*$. The \emph{language} $L(R)$ generated by $R$ is defined as follows. If $\alpha \in \Sigma \cup \{\epsilon\}$, then $L(\alpha)$ is the set containing the single string $\alpha$.  If $S$ and $T$ are regular expressions, then $L(S \cdot T) = L(S)\cdot L(T)$, that is, any string formed by the concatenation of a string in $L(S)$ with a string in $L(T)$, $L(S)|L(T) = L(S) \cup L(T)$, and $L(S^*) = \bigcup_{i \geq 0} L(S)^i$, where $L(S)^0 = \{\epsilon\}$ and $L(S)^i = L(S)^{i-1} \cdot L(S)$, for $i > 0$. The \emph{parse tree} $\PT(R)$ of $R$ (not to be confused with the parse of $Q$ wrt. to $R$) is the rooted, binary tree representing the hierarchical structure of $R$.  The leaves of $\PT(R)$ are labeled by a character from $\Sigma$ or $\epsilon$ and internal nodes are labeled by either $\cdot$, $\mid$, or $*$.

\paragraph*{Finite Automata}
 A \emph{finite automaton} is a tuple $A = (V, E, \Sigma, \theta,
\phi)$, where $V$ is a set of nodes called \emph{states}, $E$ is a set of directed edges between states called \emph{transitions} either labeled $\epsilon$ (called $\epsilon$-transitions) or labeled by a character from $\Sigma$ (called character-transitions), $\theta \in V$
is a \emph{start state}, and $\phi \in V$ is an \emph{accepting state}\footnote{Sometimes NFAs are allowed a \emph{set} of accepting states, but this is not necessary for our purposes.}. In short, $A$ is
an edge-labeled directed graph with a special start and accepting node. $A$ is a \emph{deterministic finite automaton} (DFA) if $A$ does not contain any $\epsilon$-transitions, and all outgoing
transitions of any state have different labels. Otherwise, $A$ is a \emph{non-deterministic automaton} (NFA). When we deal with multiple automatons, we use a subscript $_A$ to indicate information associated with automaton $A$, e.g., $\theta_A$ is the start state of automaton $A$.

Given a string $Q$ and a path $P$ in $A$ we say that $Q$ and $P$ \emph{match} if the concatenation of the labels on the transitions in $P$ is $Q$. Given a state $s$ in $A$ we define the \emph{state-set transition} $\delta_A(s, Q)$ to be the set of states reachable from $s$ through paths matching $Q$. For a set of states $S$ we define $\delta_A(S, Q) = \bigcup_{s \in S} \delta_A(s, Q)$. We say that $A$ \emph{accepts} the string $Q$ if $\phi_A\in \delta_A(\theta_A, Q)$. Otherwise $A$ \emph{rejects} $q$. For an accepting path $P$ in $A$, we define the \emph{parse} of $P$ for $A$ to be the sequence of character transitions in $A$ on $P$. Given a string $Q$ accepted by $A$, a parse of $Q$ is a parse for $A$ of any accepting path matching $Q$.

We can use a sequence of state-set transitions to test acceptance of a string $Q$ of length $n$ by computing a sequence of state-sets $S_0, \ldots, S_n$, given by $S_0 = \delta_A(\theta_A, \epsilon)$ and $S_i = \delta_A(S_{i-1}, Q[i])$, $i=1, \ldots, n$. We have that $\phi_A \in S_n$ iff $A$ accepts $Q$. We can  extend the algorithm to also compute the parse of $Q$ for $A$ by processing the state-sets in reverse order to recover an accepting path and output the character transitions. Note that for matching we only need to store the latest two state-sets at any point to compute the final state-set $S_n$, whereas for parsing we store the full sequence of state-sets. 
\paragraph*{Thompson NFA}
\begin{figure}[t]
  \centering \includegraphics[scale=.5]{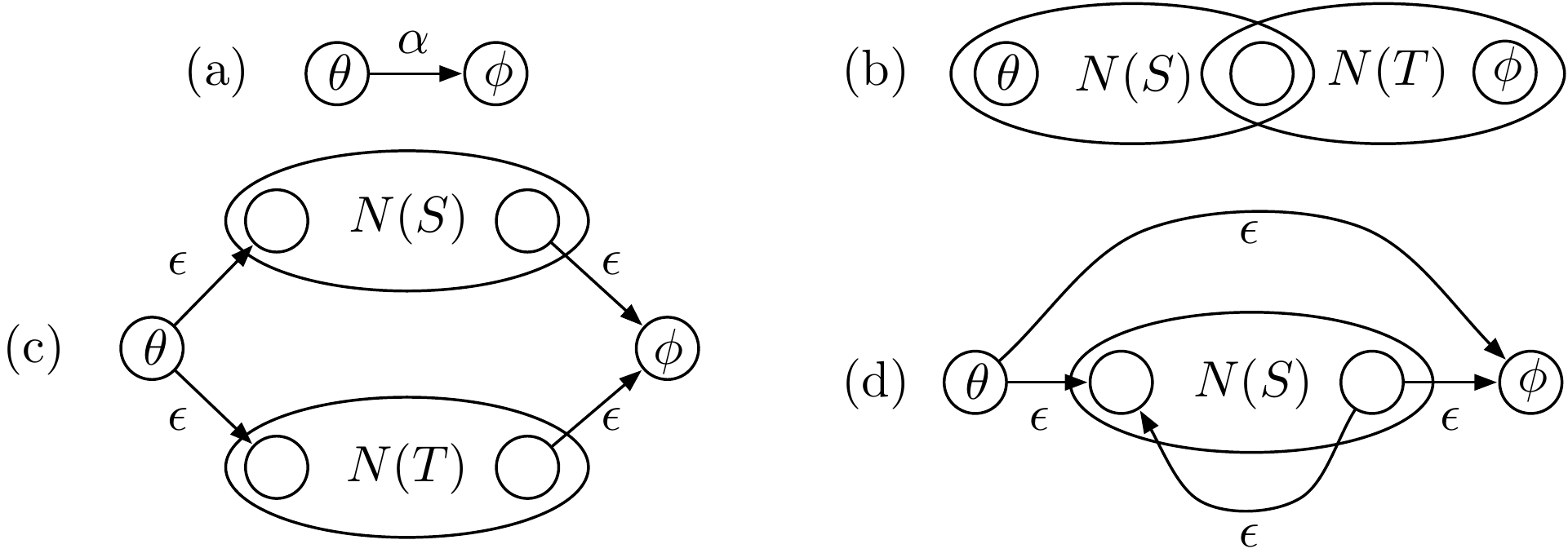}
    \caption{Thompson's recursive NFA construction. The regular
   expression $\alpha \in \Sigma \cup \{\epsilon\}$ corresponds to NFA
   $(a)$. If $S$ and $T$ are regular expressions then $N(ST)$,
   $N(S|T)$, and $N(S^*)$ correspond to NFAs $(b)$, $(c)$, and $(d)$,
   respectively.  In each of these figures, the leftmost node $\theta$
   and rightmost node $\phi$ are the start and the accept nodes,
   respectively.  For the top recursive calls, these are the start and
   accept nodes of the overall automaton. In the recursions indicated,
   e.g., for $N(ST)$ in (b), we take the start node of the
   subautomaton $N(S)$ and identify with the state immediately to the
  left of $N(S)$ in (b). Similarly the accept node of $N(S)$ is
   identified with the state immediately to the right of $N(S)$ in
   (b).}
  \label{fig:thompson}
\end{figure}

Given a regular expression $R$, we can construct an NFA accepting precisely the strings in $L(R)$ by several  classic methods~\cite{MY1960, Glushkov1961, Thomp1968}. In particular, Thompson~\cite{Thomp1968} gave the simple well-known construction shown in Fig.~\ref{fig:thompson}. We will call an NFA constructed with these rules a \emph{Thompson NFA} (TNFA). A TNFA $N(R)$ for $R$ has at most $2m$ states, at most $4m$ transitions, and can be computed in $O(m)$ time. Note that each character in $R$ corresponds to a unique character transition in $N(R)$ and hence a parse of a string $Q$ for $N(R)$ directly corresponds to a parse of $Q$ for $R$.   The parse tree of a TNFA $N(R)$ is the parse tree of $R$.  
 With a breadth-first search of $A$ we can compute a state-set transition for a single character in $O(m)$ time. By our above discussion, it follows that we can solve regular expression matching in $O(nm)$ time and $O(m)$ space, and regular expression parsing in $O(nm)$ time and $O(nm)$ space.

\paragraph*{TNFA Decomposition}
We need the following decomposition result for TFNAs (see Fig.~\ref{fig:subTNFA}). Similar decompositions are used in~\cite{Myers1992, Bille06}. Given a TNFA $A$ with $m > 2$ states, we decompose $A$ into an \emph{inner subTNFA} $A_I$ and an \emph{outer subTNFA} $A_O$. The inner subTNFA consists of a pair of \emph{boundary states} $\theta_{A_I}$ and $\phi_{A_I}$ and all states and transitions that are reachable from $\theta_{A_I}$ without going through $\phi_{A_I}$. %, but not crossing $\phi_{A_I}$. 
Furthermore, if there is a path of $\epsilon$-transitions from $\phi_{A_I}$ to $ \theta_{A_I}$ in $A_O$, we add an $\epsilon$-transition from $\phi_{A_I}$ to $ \theta_{A_I}$ in $A_I$ (following the rules from Thompson's construction). The outer subTNFA is obtained by removing all states and transitions of $A_I$ except $\theta_{A_I}$ and $\phi_{A_I}$. Between $\theta_{A_I}$ and $\phi_{A_I}$ we add a special transition labeled $\beta_{A_I} \not \in \Sigma$ and if $A_I$ accepts the empty string we also add an $\epsilon$-transition (corresponding to the regular expression $(\beta_{A_I} \mid \epsilon)$). The decomposition has the following properties. Similar results are proved  in~\cite{Myers1992, Bille06}.
\begin{figure}[t]
  \centering \includegraphics[scale=.33]{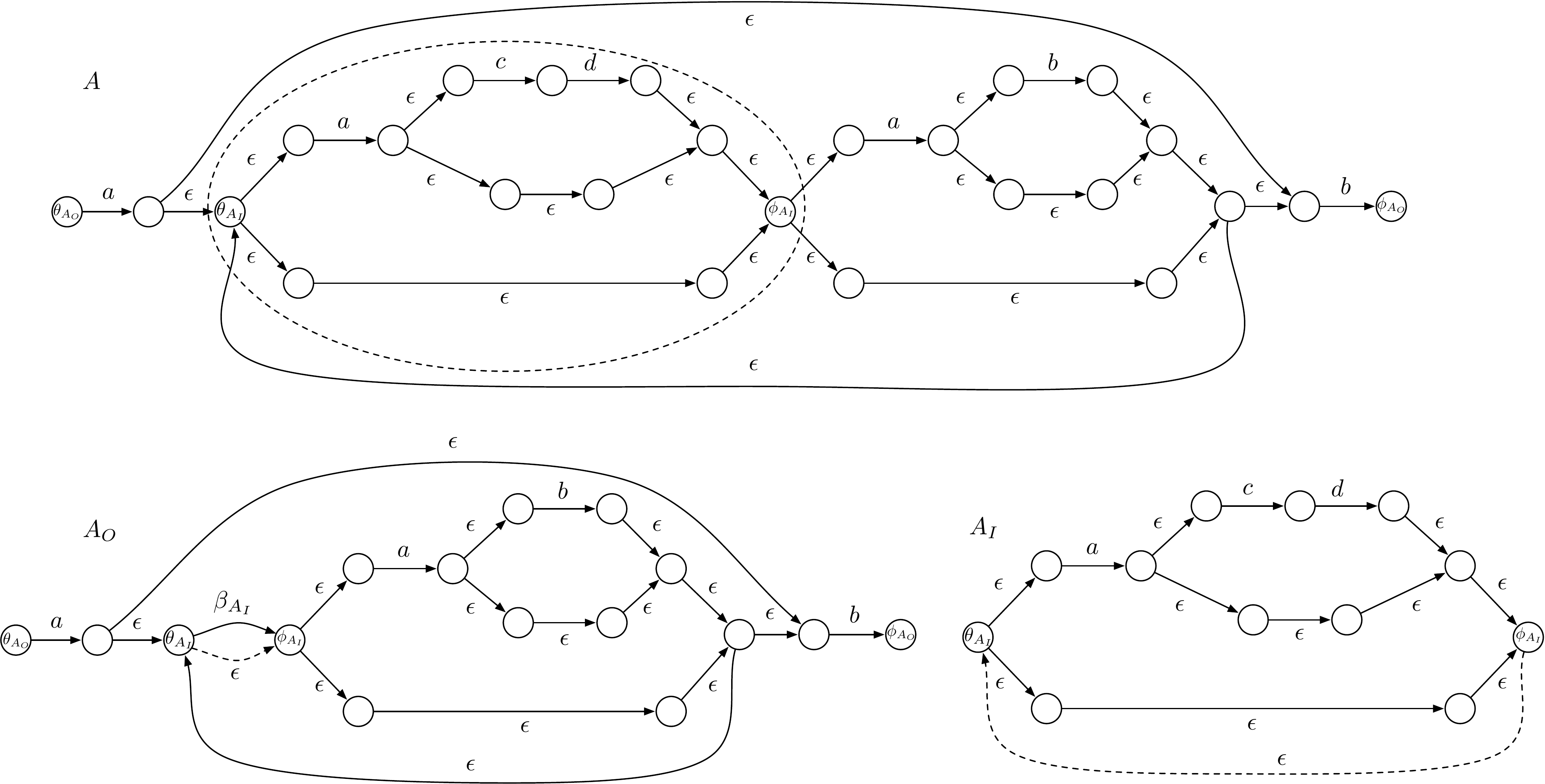}
  \caption{Decomposition of TNFA $A$ into subTNFAs $A_O$ and $A_I$. The dotted $\epsilon$-transition in $A_O$ exists since $A_I$ accepts the empty string, and the dotted $\epsilon$-transition in $A_I$ exists since there is a path of $\epsilon$-transitions from $\phi_{A_I}$ to $ \theta_{A_I}$.}\label{fig:subTNFA}
\end{figure}  

\begin{lemma}\label{lem:separator}
Let $A$ be any TNFA with $m > 2$ states. In $O(m)$ time we can decompose $A$ into inner and outer subTNFAs $\AO$ and $\AI$ such that 
\begin{itemize}
\item[(i)] $\AO$ and $\AI$ have at most $2/3m + 8$ states each, and 
\item[(ii)] any path from $\AO$ to $\AI$ crosses $\theta_{\AI}$ and any path from $\AI$ to $\AO$ crosses $\phi_{\AI}$. 
\end{itemize} 
\end{lemma}
\begin{proof}
Consider the parse tree $\PT$ of $A$ with $v$ nodes. Since $\PT$ is a binary tree with more than one node we can find a separator edge $e$ in linear time whose removal splits $\PT$ into subtrees $\PT_I$ and $\PT_O$ that each have at most $2/3v + 1$ nodes. Here, $\PT_O$ is the subtree containing the root of $\PT$. The subTNFA $A_I$ is the TNFA induced by  $\PT_I$, possibly with a new union node as root with one child being the root of $\PT_I$ and the other a leaf labeled~$\epsilon$. The subTNFA $\PT_O$ is the TNFA induced by the tree $\PT_O$, where the separator edge is replaced by an edge to either a leaf labeled $\beta_{A_I}$ or to a union-node with children labeled $\beta_{A_I}$ and $\epsilon$ in the case where $\AI$ accepts the empty string. Thus, each subTNFA are induced by a tree with at most $2/3v + 4$ nodes. Since each node correspond to two states, each subTNFA has at most $2/3m + 8$ states. 
\end{proof}

\section{String Decompositions}
Let $A$ be a TNFA decomposed into subTNFAs $\AO$ and $\AI$ and $Q$ be a string accepted by $A$. We show how to efficiently decompose $Q$ into substrings corresponding to subpaths matched in each subTNFA. The algorithm will be a key component in our recursive algorithm in the next section.

Given an accepting path $P$ in $A$, we define the \emph{path decomposition} of $P$ wrt.\ $\AI$ to be the partition of $P$ into a sequence of subpaths $\mathcal{P} = \overline{p}_1, p_1, \overline{p}_2, p_2, \ldots , \overline{p}_\ell,p_\ell, \overline{p}_{\ell+1}$, where the \emph{outer subpaths}, $\overline{p}_1, \ldots, \overline{p}_{\ell+1}$, are subpaths in $\AO$ and the \emph{inner subpaths}, $p_1, \ldots, p_{\ell}$ are the subpaths in $\AI$. The \emph{string decomposition} induced by $\mathcal{P}$ is the sequence of substrings $\mathcal{Q} = \overline{q}_1, q_1, \overline{q}_2, q_2, \ldots , \overline{q}_\ell,q_\ell, \overline{q}_{\ell+1}$ formed by concatenating the labels of the corresponding subpath in $A$. A sequence of substrings is a substring decomposition wrt.\ to $\AI$ if there exists an accepting path that induces it. Our goal is to compute a string decomposition in $O(nm)$ time and $O(n + m)$ space, where $n$ is the length of $Q$ and $m$ is the number of states in $A$. 

An immediate idea would be to process $Q$ from left to right using state-set transitions and "collapse" the state set to a boundary state $b$ of $\AI$ whenever the state set contains $b$ and there is a path from $b$ to $\phi_A$ matching the rest of $Q$. Since $\AO$ and $\AI$ only interacts at the boundary states, this effectively corresponds to alternating the simulation of $A$ between $\AO$ and $\AI$. However, because of potential cyclic dependencies from paths of $\epsilon$-transition from $\phi_{\AI}$ to $\theta_{\AI}$ in $\AO$ and $\theta_{\AI}$ to $\phi_{\AI}$ in $\AI$ we cannot immediately determine which subTNFA we should proceed in and hence we cannot correctly compute the string decomposition.  
For instance, consider the string $Q = aaacdaabaacdacdaabab$ from Figure~\ref{fig:stringdecomp}. After processing the first two characters ($aa$) both $\theta_{\AI}$ and $\phi_{\AI}$ are in the state set, and there is a path from both these states to $\phi_A$ matching the rest of $Q$. The same is true after processing the first six characters ($aaacda$). In the first case the substring consisting of the next three characters ($acd$) only matches a path in $A_I$, whereas in the second case the substring consisting of the next two characters ($ab$) only matches a path in $A_O$. A technical contribution in our algorithm in the next section is to efficiently overcome these issues by a two-step approach that first decomposes the string into substrings and labels the substrings greedily to find a correct string decomposition. 

\subsection{Computing String Decompositions}
We need the following new definitions. Let $i$ be a position in $Q$ and let $s$ be a state in $A$. We say that $(i, s)$ is a \emph{valid pair} if there is a path from $\theta_A$ to $s$ matching $Q[1, i]$ and from $s$ to $\phi_A$ matching $Q[i+1, n]$. For any set of states $X$ in $A$, we say that $(i, X)$ is a valid pair if each pair $(i, x)$, $x \in X$, is a valid pair. A path in $A$ \emph{intersects} a valid pair $(i, X)$ if $(i,x)$ for some state $x \in X$ is on the path.

Our algorithm consist of the following steps. In step 1, we process $Q$ from left to right and right to left to compute and store the \emph{match sets}, consisting of all valid pairs for the boundary states $\theta_{\AI}$ and $\phi_{\AI}$. We then use the match sets in step 2 to process $Q$ from left to right to build a sequence of valid pairs for the boundary states that all intersect a single accepting path $P$ in $A$ matching $Q$, and that has the property that all positions where the accepting path $P$ contains  $\theta_{\AI}$ or  $\phi_{\AI}$ correspond to a valid pair in the sequence. Finally, in step 3 we construct the string decomposition using a greedy labeling of the sequence of valid pairs. See Figure~\ref{fig:stringdecomp} for an example of the computation in each step.

\paragraph*{Step 1: Computing Match Sets}
First, compute the \emph{match sets} given by
\begin{align*}
\qquad \qquad \Match{\theta_{\AI}} &= \{i \mid (i, \theta_{\AI}) \text{ is a valid pair}\} \\
\Match{\phi_{\AI}} &= \{i \mid (i, \phi_{\AI}) \text{ is a valid pair}\} 
\end{align*}

Thus, $\Match{\theta_{\AI}}$ and $\Match{\phi_{\AI}}$ are the positions in $Q$ that correspond to a valid pair for the  boundary states $\theta_{\AI}$ and $\phi_{\AI}$, respectively. To compute these, we first compute the \emph{prefix match sets}, $\Prefix(s)$,  and \emph{suffix match sets}, $\Suffix(s)$, for $s\in\{\theta_{\AI}, \phi_{\AI}\}$. A position $i$ is in $\Prefix(s)$ if there is a path from $\theta_A$ to $s$ accepting the prefix $Q[1,i]$, and in $\Suffix(s)$ if there is a path from $s$ to $\phi_A$ accepting the suffix $Q[i+1,n]$. To compute the prefix match sets we perform state-set transitions on $Q$  and $A$ and whenever the current state-set contains either $\theta_{\AI}$ or $\phi_{\AI}$ we add the corresponding position to $\Prefix(s)$. We compute the suffix match sets in the same way, but now we perform the state-set transitions on $Q$ from right to left and $A$ with the direction of all transitions reversed. 
Each step of the state-set transition takes $O(m)$ time and hence we use $O(nm)$ time in total. 

\begin{figure}[t]
  \centering 
  \includegraphics[scale=.7]{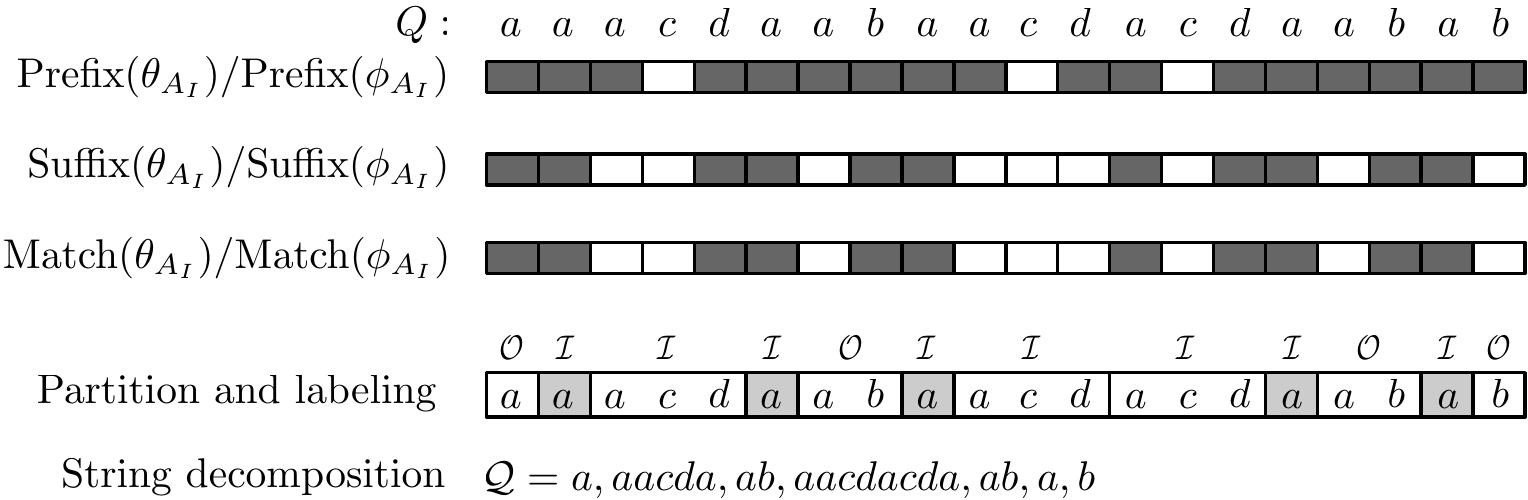}
    \caption{The string decomposition of the string $Q= aaacdaabaacdacdaabab$ wrt.\ $A_I$ in Figure~\ref{fig:subTNFA} and the corresponding suffix/prefix match sets. The dark grey blocks in the prefix, suffix and match sets are the positions contained in the sets. The blocks in the partition of the string are labeled $\mathcal{O}$ and $\mathcal{I}$ for outer and inner, respectively. The grey blocks in the partition are the substrings that can be parsed by both the inner and outer automaton. According to our procedure these blocks are labeled \emph{inner}.}
  \label{fig:stringdecomp}
\end{figure}

Finally, we compute the match sets $\Match{s}$, for $s\in\{\theta_{\AI},\phi_{\AI}\}$, by taking the intersection of $\Prefix(s)$ and $\Suffix(s)$. In total we use $O(mn)$ time and $O(n+m)$ space to compute and store the match sets. 

\paragraph*{Step 2: Computing Valid Pairs} 
We now compute a sequence of valid pairs 
$$V = (i_1, X_1), (i_2, X_2), \ldots, (i_{k}, X_{k})$$ such that $0 \leq i_1 < \cdots < i_{k} \leq n$ and $X_j \subseteq \{\theta_{\AI}, \phi_{\AI}\}$ and with the property that the states of all pairs intersect a single accepting path $P$ in $A$ and at all places where $P$ is equal to either $\theta_{\AI}$ or $\phi_{\AI}$ correspond to a valid pair in $V$.

To compute the sequence we run a slightly modified state-set transition algorithm:  %Initially, set $S_0 := \delta_A(\theta_A,\emptyset)$.
For $i = 0,1, \ldots, n$ we set $S_i = \delta_A(S_{i-1}, Q[i])$ (for $i=0$ set $S_0 := \delta_A(\theta_A,\epsilon)$) and compute the set 
$$ X:= \{x \mid x\in\{\theta_{\AI}, \phi_{\AI}\} \textrm{ and } i \in \Match{x}\} \cap S_i\;.$$ 
Thus $X$ is the set of boundary nodes in $S_i$ that corresponds to a valid pair computed in Step~1. 
If $X\neq \emptyset$ we add $(i,X)$ to the sequence $V$ and set $S_i := X$.

\medskip

We argue that this produces a sequence $V$ of valid pairs with the required properties. First note that by definition of $X$ we inductively produce state-set $S_0, \ldots, S_n$ such that $S_i$ contains the set of states reachable from $\theta_{A}$ that match $Q[1,i]$ and the paths used to reach $S_i$ intersect the states of the valid pairs produced up to this point. Furthermore, we include all positions in $V$ where $S_i$ contains $\theta_{\AI}$ or $\phi_{\AI}$. It follows that $V$ satisfies the properties. 

Each of modified state-set transition uses $O(m)$ time and hence we use $O(nm)$ time in total. The sequence $V$ uses $O(n)$ space. In addition to this we store the match sets and a constant number of active state-sets using $O(n + m)$ space.

\paragraph*{Step 3: Computing the String Decomposition} We now convert the sequence
$V = (i_1, X_1), (i_2, X_2), \ldots, (i_{k}, X_{k})$  into a string decomposition. The sequence $V$ induces a partitioning $q_0, \ldots, q_{k+1}$ of $Q$ where $q_0 = Q[1,i_1]$, $q_j = Q[i_j + 1, i_{j+1}]$, and $q_{k+1} = Q[i_k+1, n]$. Note that $q_0$ and $q_{k+1}$ may be the empty string.

We now show how to efficiently convert this partition into a string decomposition. First, we determine for each of the substrings whether it can be an inner substring, an outer substring, or both. 

\subparagraph*{Labeling}  We label the substrings as follows. First label $q_0$ and $q_{k+1}$ with outer. For the rest of the substrings, if $X_i = \{\theta_{\AI}\}$ and $X_{i+1} = \{\phi_{\AI}\}$ then label $q_i$ with inner, and if $X_i = \{\phi_{\AI}\}$ and $X_{i+1} = \{\theta_{\AI}\}$  then label $q_i$ with outer. If either $X_{i}$ and $X_{i+1}$ contains more than one boundary node then we use standard state-set transitions in $\AI$ and $\AO$ to determine if $\AI$ accepts $q_i$ or if there is a path in $\AO$ from $\phi_{\AI}$ to $\theta_{\AI}$ that matches $q_i$.  If a substring is only accepted by $\AI$ or $\AO$ it is labeled with inner or outer, respectively.
Let $q_i$ be a substring that can be both an inner or an outer substring.
We divide this into two cases. If there is an $\epsilon$-path from $\phi_{\AI}$ to $\theta_{\AI}$ then label all such $q_i$ with inner. Otherwise %(there is an $\epsilon$-path from $\theta_{\AI}$ to $\phi_{\AI}$, but not from $\phi_{\AI}$ to $\theta_{\AI}$) 
label all such $q_i$ with outer. See also Algorithm~\ref{algo:label}.

\begin{algorithm}[t]
   \DontPrintSemicolon
   \KwIn{A sequence $V$ of valid pairs $(i_1,X_1),\ldots(i_k,V_k)$ and the corresponding partition $q_0,\ldots, q_{k+1}$ of $Q$.}
   \KwOut{A labeling of the partition}
   
   The (possible empty) substrings $q_0$ and $q_{k+1}$ are labeled outer.
   
   \For{$i=1$ \KwTo $k$}{
    \uIf(\tcc*[f]{Case 1}){$X_i = \{\theta_{\AI}\}$ and $X_{i+1} = \{\phi_{\AI}\}$}{label $q_i$ inner.}
    \uElseIf(\tcc*[f]{Case 2}){$X_i = \{\phi_{\AI}\}$ and $X_{i+1} = \{\theta_{\AI}\}$}{ 
    label $q_i$ outer.}
    
    \ElseIf(\tcc*[f]{Case 3}){ $X_{i}$ or $X_{i+1}$ contains more than one boundary node}{Use standard state-set transitions in $\AI$ and $\AO$ to determine if $\AI$ accepts $q_i$ or if there is a path in $\AO$ from $\phi_{\AI}$ to $\theta_{\AI}$ that matches $q_i$.
    
    \uIf(\tcc*[f]{Subcase 3a}){$q_i$ is only accepted by $\AI$} {label $q_i$ inner}
    \uElseIf(\tcc*[f]{Subcase 3b}){$q_i$ is only accepted by $\AO$} {label $q_i$ with outer}
    
    \Else(\tcc*[h]{$q_i$ is accepted by both $\AI$ and $\AO$} ){There are two cases.
        
        \uIf(\tcc*[f]{Subcase 3c}){there is an $\epsilon$-path from $\phi_{\AI}$ to $\theta_{\AI}$}{label $q_i$ inner.}
        \lElse(\tcc*[f]{Subcase 3d}){label $q_i$ outer.}
  
        }}}    
\caption{Labeling}
\label{algo:label}
\end{algorithm}

For correctness first note that $q_0$ and $q_{k+1}$ are always (possibly empty) outer substrings. 
The cases where both $|X_i| = |X_{i+1}|$ (case 1 and 2) are correct by the correctness of the sequence of valid pairs $V$.
Due to cyclic dependencies we may have that $X_{i}$ and $X_{i+1}$ contains more than one boundary node. This can happen if there is an $\epsilon$-path from $\theta_{\AI}$ to $\phi_{\AI}$ and/or there is an $\epsilon$-path from $\phi_{\AI}$ to $\theta_{\AI}$. If a substring only is accepted by one of $\AI$ (case 3a) or $\AO$ (case 3b) then it follows from the correctness of $V$ that the labeling is correct. 
It remains to argue that the labeling in the case where $q_i$ is accepted by both $\AI$ and $\AO$ is correct.
To see why the labeling in this case is consistent with a string decomposition of the accepting path, note that if there is an $\epsilon$-path from $\phi_{\AI}$ to $\theta_{\AI}$ (case 3c), then it is safe to label $q_i$ with inner, since in the case where 
we are in $\phi_{A_I}$ after having read $q_{i-1}$  we can just follow the $\epsilon$-path to $\theta_{A_I}$ and then start reading $q_i$ from here. The argument for the other case (case 3d) is similar. 

Except for the state-set transitions in case 3 all cases takes constant time. The total time of all the state-set transitions ia $O(nm)$. The space of $V$ and the partition together with the labeling uses $O(n)$ space.

\subparagraph*{String decomposition}
Now every substring has a label that is either inner or outer. We then merge adjacent substrings that have the same label. 
This produces an alternating sequence of inner and outer substrings, which is the final string decomposition. Such an alternating subsequence must always exist since each pair in $V$ intersects an accepting path. 
 
\medskip

In summary, we have the following result. 
\begin{lemma}\label{lem:stringdecomposition}
Given string $Q$ of length $n$, and TNFA $A$ with $m$ states decomposed into $\AO$ and $\AI$, we can compute a string decomposition wrt. $\AI$ in $O(nm)$ time and $O(n+m)$ space. 
\end{lemma}

\section{Computing Accepting Paths}\label{sec:acceptingpaths}
Let $Q$ be a string of length $n$ accepted by a TNFA $A$ with $m$ states. We now show how to compute an accepting path for $Q$ in $A$ in $O(nm)$ time and $O(n + m)$ states.

Since an accepting path may have length $\Omega(nm)$ (there may be $\Omega(m)$ $\epsilon$-transitions between each character transition) we cannot afford to explicitly compute the path in our later $o(nm)$ time algorithms. Instead, our algorithm will compute  \emph{compressed paths} of size $O(n)$ obtained by deleting all $\epsilon$-transitions from a path. Note that the compressed path stores precisely the information needed to solve the parsing problem. 

\paragraph*{Computing Compressed Paths}
\subparagraph*{Algorithm $\Path(A,Q)$}  We define a recursive algorithm $\Path(A,Q)$ that takes a TNFA $A$ and a string $Q$ and returns a compressed accepting path in $A$ matching $Q$ as follows. If $n < \gamma_n$ or $m < \gamma_m$, for constants $\gamma_n, \gamma_m > 0$ that we will fix later, we simply run the naive algorithm that stores all state-sets during state-set simulation from left-to-right in $Q$. Since one of $n$ or $m$ is constant this uses $O(nm) = O(n + m)$ time and space. Otherwise, we proceed according to following steps.

\subparagraph*{Step 1: Decompose} We compute a decomposition of $A$ into inner and outer subTNFAs $\AI$ and $\AO$ and compute a corresponding string decomposition $\mathcal{Q} = \overline{q}_1, q_1, \overline{q}_2, q_2, \overline{q}_\ell, q_\ell, \overline{q}_{\ell+1}$ for $Q$.

\subparagraph*{Step 2: Recurse}
We build a single substring corresponding to all the subpaths in $\AO$ and $\ell$ substrings for $\AI$ (one for each subpath in $\AI)$ and recursively compute the compressed paths. To do so, construct $\overline{q} = \overline{q}_1\cdot \beta_{A_I} \cdot \overline{q}_2 \cdot \beta_{A_I} \cdots \beta_{A_I} \cdot \overline{q}_{\ell+1}$. Recall, that $\beta_{A_I}$ is the label of the special transition we added between $\theta_{A_I}$ and $\phi_{A_I}$in $A_O$.  %(see Fig.~\ref{fig:stringdecomp}). 
Then, compute the compressed paths 
\begin{align*}
\qquad \qquad \overline{p} &= \Path(\AO, \overline{q}) \\
p_i &= \Path(A_I, q_i) \qquad 1\leq i \leq \ell \\
\end{align*}

\subparagraph*{Step 3: Combine}
Finally, extract the subpaths $\overline{p}_1, \overline{p}_2, \ldots, \overline{p}_{\ell+1}$ from $\overline{p}$ corresponding to the substrings $\overline{q}_1, \overline{q}_2, \ldots,  \overline{q}_{\ell+1}$ and return the combined compressed path 
$$P = \overline{p}_0 \cdot p_1 \cdot \overline{p}_1 \cdot p_2 \cdot\overline{p}_2 \cdots p_\ell \cdot\overline{p}_{\ell+1}
$$

Inductively, it directly follows that the returned compressed path is a compressed accepting path for $Q$ in $A$.

\subsection{Analysis}\label{sec:analysis-rec}
We now show that the total time $T(n,m)$ of the algorithm is $O(nm)$. If $n < \gamma_n$ or $m < \gamma_m$, we run the backtracking algorithm using $O(nm) = O(n + m)$ time and space. If $n \geq \gamma_n$ and $m \geq \gamma_m$, we implement the recursive step of the algorithm using $O(nm)$ time. Let $n_i$ be the length of the inner string $q_i$ in the string decomposition and let $n_0 = \sum_{i=1}^{\ell+1}|\bar{q}_i|$. Thus, $n= \sum_{i=1}^{\ell+1}n_i$ and $|\bar{q}|=n_0+\ell$. In step 2, the recursive call to compute $\overline{p}$ takes $O(T(n_0 +\ell,\frac{2}{3}m+8))$ time and the recursive calls to compute $p_1, \ldots, p_\ell$ takes $\sum_{i=1}^\ell T(n_i, \frac{2}{3}m+8)$. The remaining steps of the algorithm all takes $O(nm)$ time. Hence, we have the following recurrence for $T(n,m)$.
$$
T(n,m) = 
\begin{cases}
\sum_{i=1}^\ell T(n_i, \frac{2}{3}m+8)+ T(n_0 +\ell,\frac{2}{3}m+8) + O(mn) & m \geq \gamma_m \textrm{ and } n\geq \gamma_n  \\
 O(m+n) & m < \gamma_m  
 \textrm{ or }  n  <  \gamma_n % m = \Theta(x), n = \Theta(y)
\end{cases}
$$
It follows that $T(n,m)=O(nm)$ for $\gamma_n=2$ and $\gamma_m=25$. See Appendix~\ref{sec:recurrenceproof} for a detailed proof. 

Next, we consider the space complexity. First, note that the total space for storing $R$ and $Q$ is $O(n + m)$. To analyse the space during the recursive calls of the algorithm, consider the recursion tree $\RT$ for $\Path(A, Q)$. For a node $v$ in $\RT$, we define $Q_v$ of length $n_v$ to be the string and $A_v$ with $m_v$ states to be the TNFA at $v$. 
Consider a path $v_1, \ldots, v_j$ of nodes in $\RT$ from the root to leaf $v_j$ corresponding to a sequence of nested recursive calls of the algorithm. If we naively store the subTNFAs, the string decompositions, and the compressed paths, we use $O(n_{v_i} + m_{v_i})$ space at each node $v_i$, $1 \leq i \leq j$. By Lemma~\ref{lem:separator}(i) the sum of the sizes of the subTNFAs on a path forms a geometrically decreasing sequence and hence the total space for the subTNFAs is  $\sum_{i=1}^j m_{v_i}=O(m)$. However, since each string (and hence compressed path) at each node $v_i$, $1 \leq i \leq j$, may have length $\Omega(n)$ we may need $\Omega(n \log m)$ space in total for these. We show how to improve this to $O(n + m)$ space in the next section.

\subsection{Squeezing into Linear Space}
We now show how to improve the space to $O(n + m)$. To do so we show how to carefully implement the recursion to only store the strings for a selected subset of the nodes along any root to leaf path in $\RT$ that in total take no more than $O(n)$ space. 

First, consider a node $v$ in $\RT$ and the corresponding string $Q_v$ and TNFA $A_v$. Define $\chi^Q_v$ to be the function that maps each character position in $Q_v$ (ignoring $\beta_{\AI}$ transitions) to the unique corresponding character in $Q$ and $\chi^A_v$ to be the function that maps each character transition (non-$\epsilon$ transition) in $A_v$ to the unique character transition in $A$. Note that these are well-defined by the construction of subproblems in the algorithm. At a node $v$, we represent $\chi^Q_v$ by storing for each character in $Q_v$ a pointer to the corresponding character in $Q$. Similarly, we represent $\chi^A_v$ by storing for each character transition in $A_v$ a pointer to the corresponding character transition in $A$. This uses $O(n_v + m_v)$ additional space. It is straightforward to compute these mappings during the algorithm directly from the string and TNFA decomposition in the same time. With the mappings we can now output transitions on the compressed path as pairs of position in $Q$ and transitions in $A$ immediately as they are computed in a leaf of $\RT$. Thus, when we have traversed a full subtree at a node we can free the space for that node since we do not have to wait to return to the root to combine the pieces of the subpath with other subpaths. 

We combine the mappings with an ordering of the recursion according to a \emph{heavy-path decomposition} of $\RT$. Let $v$ be an internal node in $\RT$ and define the \emph{heavy child} of $v$ to be a child of $v$ of maximum string length among the children of $v$. The remaining children are \emph{light children} of $v$. We have the following key property of light children. 

\begin{lemma}\label{lem:heavypath}
	For any node $v$ with light child $u$ in $\RT$, we have that $n_u \leq \frac{3}{4} n_v + O(1)$. 
\end{lemma}

\begin{proof}
Let $v$ be a node in $\RT$ with $\ell+1$ children. The total string length of the children of $v$ is $n_v+\ell$ and a light child of $v$ can have string length at most $(n_v + \ell)/2$. Since the $\ell$ inner strings are disjoint non-empty substrings of $Q$ separated by at least one character (the non-empty outer substrings), we have that $\ell \leq (n_v+1)/2$. Hence, a light child $u$ can have string length at most $n_u \leq \frac{n_v + \ell}{2} < \frac{3}{4} n_v + 1$.
\end{proof}

We order the recursive calls at each node $v$ as follows. First, we recursively visit all the light children of $v$ and upon completing each recursive call, we free the space for that node. Note that the mappings allow us to do this. We then construct the subproblem for the heavy child of $v$, free the space for $v$, and continue the recursion at the heavy child. 

To analyse the space of the modified algorithm, consider a path a path $v_1, \ldots, v_j$ of nodes in $\RT$ from the root to a leaf $v_j$. We now have that only nodes $v_i$, $1 \leq i < j$ will explicitly store a string if $v_{i+1}$ is a light child of $v_i$. By Lemma~\ref{lem:heavypath} the sum of these lengths form a geometrically decreasing sequence and hence the total space is now $O(n)$. In summary, we have shown the following result.

\begin{theorem}
	Given a TNFA with $m$ states and a string of length $n$, we can compute a compressed accepting path for $Q$ in $A$ in $O(nm)$ time and $O(n+m)$ space. 
\end{theorem}
Note that the algorithm works in a comparison-based, pointer model of computation. By our discussion this immediately implies the main result of Thm.~\ref{thm:main}.

%!TEX root = paper.tex
\section{Speeding up the Algorithm}
We now show how to adapt the algorithm to use the faster state-set simulations algorithms such as Myers' algorithm~\cite{Myers1992} and later variants~\cite{BFC2008,Bille06} that improve the $O(m)$ bound for a single state-set transition. These results and ours all work on a unit-cost word RAM model of computation with $w$-bit words and a standard instruction set including addition, bitwise boolean operations, shifts, and multiplication. We can store a pointer to any position in the input and hence $w \geq \log (n + m)$. For simplicity, we will show how to adapt the tabulation-based algorithm of Bille and Farach-Colton~\cite{BFC2008}.

\subsection{Fast Matching}
Let $A$ be a TNFA with $m$ states and let $Q$ be a string of length $n$. Assume first that the alphabet is constant. We briefly describe the main features of the algorithm by Bille and Farach-Colton~\cite{BFC2008} that solves the matching problem in $O(nm/\log n + n + m)$ time and $O(n^{\varepsilon} + m)$ space, for any constant $\varepsilon > 0$. In the next section we show how to adapt the algorithm to compute an accepting path.

Given a parameter $t < w$, we will construct a global table of size $2^{ct} < 2^w$, for a constant $c$, to speed up the state-set transition. We decompose $A$ into a tree $MS$ of $O(\ceil{m/t})$ micro TNFAs, each with at most $t$ states. For each $M \in MS$, each child $C$ is represented by its start and accepting state and a pseudo-transition connecting these. By similar arguments as in Lemma~\ref{lem:separator} we can always construct such a decomposition.

We represent each micro TNFA $M \in MS$ uniquely by its parse tree using $O(t)$ bits. Since $M$ has at most $t$ states, we can represent the state-set for $M$, called the \emph{local state-set} and denoted $S_M$, using $t$ bits. Hence, we can make a universal table of size $2^{O(t)}$ that for every possible micro TNFA $M$ of size $\leq t$, local state-set $S_M$, and character $\alpha \in \Sigma \cup \{\epsilon\}$ computes the state-set transition $\delta_M(S_M, \alpha)$ in constant time.

We use the tabulation to efficiently implement a global state-set transition on $A$ as follows. We represent the state-set for $A$ as the union of the local state-sets in $MS$. Note that parents and children in $MS$ share some local states, and these states will be copied in the local state-sets.

To implement a state-set transition on $A$ for a character $\alpha$, we first traverse all transitions labeled $\alpha$ in each micro TNFA from the current state-set. We then follow paths of $\epsilon$ transition in two depth-first left-to-right traversal of $MS$. At each micro TNFA $M$, we compute all states reachable via $\epsilon$-transitions and propagate the shared states among parents and children in $MS$. Since any cycle free path in a TNFA contains at most one \emph{back transition} (see  ~\cite[Lemma 1]{Myers1992}) it follows that two such traversals suffices to to correctly compute all states in $A$ reachable via $\epsilon$-transitions. 

With the universal table, we process each micro TNFA in constant time, leading to an algorithm using $O(|MS|/t + n + m) = O(nm/t + n + m)$ time and $O(2^t + m)$ space. Setting $t = \varepsilon \log n$ produces the stated result. Note that each state-set uses $O(\ceil{m/t})$ space. To handle general alphabets, we store dictionaries for each micro TNFA with bit masks corresponding to characters appearing in the TNFA and combine these with an additional masking step in state-set transition. The leads a general solution with the same time and space bounds as above.\footnote{Note that the time bound in the original paper has an additional $m \log m$ term~\cite{BFC2008}. Using atomic heaps~\cite{fredmanwillardsp-n-mst} to represent dictionaries for micro TNFAs this term is straightforward to improve to $O(m)$. See also Bille and Thorup~\cite[Appendix A]{BT2009}.} 

\subsection{Fast Parsing}
We now show how to modify our algorithm from Sec.~\ref{sec:acceptingpaths} to take advantage of the fast state-set transition algorithm.  Let $t < w$ be the desired speed up as above. We have the following two cases. 

If $n \geq t$ and $m \geq t$ we implement the recursive step of the algorithm but replace all state-set transitions, that is when we compute the match sets and valid pairs, by the fast state-set transition algorithm. To compute the suffix match sets we need to compute fast state-set transitions on $A$ with the direction of all transitions reversed. To do so, we make a new universal table of size $2^{O(t)}$ for the micro TNFAs with the direction of transitions reversed.  We traverse the tree of micro TNFAs with two depth traversals as before except that we now traverse children in right to left order to correctly compute all reachable states. It follows that this uses $O(nm/t)$ time.

Otherwise ($n < t$ or $m < t$), we use backtracking to compute the accepting path as follows. First, we process $Q$ from left-to-right using fast state-set transitions to compute the sets $S_0, \ldots, S_n$ of states reachable via paths from $\theta_A$ for each prefix of $Q$. We store each of these state-sets. This uses $O(nm/t + n + m) = O(n + m)$ time and space. Then, we process $Q$ from right-to-left to recover a compressed accepting path in $A$. Starting from $\phi_A$ we repeatedly do a fast state-set transition $A$ with the direction of transition reversed, compute the intersection of the resulting state-set with the corresponding state-set from the first step, and update the state-set to a state in the intersection. We can compute the intersection of state-sets and the update in $O(m/t)$ time using standard bit wise operations. We do the state-set transitions on the TNFA with directions of transitions reversed as above. In total, this uses $O(nm/t + n + m) = O(n + m)$ time. 

In summary, we have the following recurrence for the time $T(n,m)$. 
\begin{equation*}
T(n,m) = 
\begin{cases}
	\sum_{i=0}^\ell T(n_i, \frac{2}{3}m + 8) + T(n_0 + \ell, \frac{2}{3}m + 8) +  O\left(\frac{nm}{t}\right) & \text{$n \geq t$ and $m \geq t$} \\
	O\left(n + m \right) & \text{$m < t$ or $n < t$}
\end{cases}
\end{equation*}
It's straightforward to show that $T(n,m) = O(nm/t + n + m)$, for $25 \leq t < w$. The space is linear as before. Plugging in $t = \epsilon \log n$ and including the preprocessing time and space for the universal tables we obtain the following logarithmic improvement of Theorem~\ref{thm:main}. 

\begin{theorem}
  	Given a regular expression of length $m$, a string of length $n$, we can solve the regular expression parsing problem in $O(nm/\log n + n + m)$ time and $O(n+m)$ space. 
\end{theorem}

Other fast state-set transition algorithms~\cite{Myers1992, Bille06} are straightforward to adapt using the above framework. The key requirement (satisfied by the existing solution) is that the algorithms need to efficiently support state-set transitions on the TNFA with the directions of the transitions reversed and intersection between two state-sets.

An algorithm that does not appear to work within the above framework is the $O(\frac{nm\log \log n}{\log^{3/2} n})$ matching algorithm by Bille and Thorup~\cite{BT2009}. This algorithm does not produce state-sets at each character in the string, but instead  maintains specialized information to derive state-sets every $\sqrt{\log n}$ characters,  which we do not see how to maintain in our parsing framework.

\bibliography{paper}

\begin{thebibliography}{10}

\bibitem{ASU1986}
A.~V. Aho, R.~Sethi, and J.~D. Ullman.
\newblock {\em Compilers: principles, techniques, and tools}.
\newblock Addison-Wesley Longman Publishing Co., Inc., 1986.

\bibitem{BI2016}
A.~Backurs and P.~Indyk.
\newblock Which regular expression patterns are hard to match?
\newblock In {\em Proc. 57th FOCS}, pages 457--466, 2016.

\bibitem{Bille06}
P.~Bille.
\newblock New algorithms for regular expression matching.
\newblock In {\em Proc. of the 33rd ICALP}, pages 643--654, 2006.

\bibitem{BFC2008}
P.~Bille and M.~Farach{-}Colton.
\newblock Fast and compact regular expression matching.
\newblock {\em Theor. Comput. Sci.}, 409(3):486--496, 2008.

\bibitem{BT2009}
P.~Bille and M.~Thorup.
\newblock Faster regular expression matching.
\newblock In {\em Proc. 36th ICALP}, pages 171--182, 2009.
\newblock Full version with appendix available at
  \texttt{http://www2.compute.dtu.dk/~phbi/files/publications/2009fremC.pdf}.

\bibitem{BT2010}
P.~Bille and M.~Thorup.
\newblock Regular expression matching with multi-strings and intervals.
\newblock In {\em Proc. 21st SODA}, pages 1297--1308, 2010.

\bibitem{DF2000}
D.~Dub{\'e} and M.~Feeley.
\newblock Efficiently building a parse tree from a regular expression.
\newblock {\em Acta Informatica}, 37(2):121--144, 2000.

\bibitem{fredmanwillardsp-n-mst}
M.~L. Fredman and D.~E. Willard.
\newblock Trans-dichotomous algorithms for minimum spanning trees and shortest
  paths.
\newblock {\em J. Comput. System Sci.}, 48(3):533--551, 1994.

\bibitem{FC2004}
A.~Frisch and L.~Cardelli.
\newblock Greedy regular expression matching.
\newblock In {\em Proc. 31st ICALP}, volume 3142, pages 618--629, 2004.

\bibitem{GRS1999}
M.~N. Garofalakis, R.~Rastogi, and K.~Shim.
\newblock {SPIRIT}: Sequential pattern mining with regular expression
  constraints.
\newblock In {\em Proc. 25th VLDB}, pages 223--234, 1999.

\bibitem{Glushkov1961}
V.~M. Glushkov.
\newblock The abstract theory of automata.
\newblock {\em Russian Math. Surveys}, 16(5):1--53, 1961.

\bibitem{Hirschberg1975}
D.~S. Hirschberg.
\newblock A linear space algorithm for computing maximal common subsequences.
\newblock {\em Commun. ACM}, 18(6):341--343, 1975.

\bibitem{JMR2007}
T.~Johnson, S.~Muthukrishnan, and I.~Rozenbaum.
\newblock Monitoring regular expressions on out-of-order streams.
\newblock In {\em Proc. 23nd ICDE}, pages 1315--1319, 2007.

\bibitem{kearns1991}
S.~M. Kearns.
\newblock Extending regular expressions with context operators and parse
  extraction.
\newblock {\em Software: Practice and Experience}, 21(8):787--804, 1991.

\bibitem{KHDA2012}
K.~Kin, B.~Hartmann, T.~DeRose, and M.~Agrawala.
\newblock Proton: multitouch gestures as regular expressions.
\newblock In {\em Proc. SIGCHI}, pages 2885--2894, 2012.

\bibitem{KDYCT2006}
S.~Kumar, S.~Dharmapurikar, F.~Yu, P.~Crowley, and J.~Turner.
\newblock Algorithms to accelerate multiple regular expressions matching for
  deep packet inspection.
\newblock In {\em Proc. SIGCOMM}, pages 339--350, 2006.

\bibitem{Laurikari2000}
V.~Laurikari.
\newblock Nfas with tagged transitions, their conversion to deterministic
  automata and application to regular expressions.
\newblock In {\em Proc. 7th SPIRE}, pages 181--187, 2000.

\bibitem{LM2001}
Q.~Li and B.~Moon.
\newblock Indexing and querying {XML} data for regular path expressions.
\newblock In {\em Proc. 27th VLDB}, pages 361--370, 2001.

\bibitem{MY1960}
R.~McNaughton and H.~Yamada.
\newblock Regular expressions and state graphs for automata.
\newblock {\em IRE Trans. on Electronic Computers}, 9(1):39--47, 1960.

\bibitem{Murata2001}
M.~Murata.
\newblock Extended path expressions of {XML}.
\newblock In {\em Proc. 20th PODS}, pages 126--137, 2001.

\bibitem{Myers1992}
E.~W. Myers.
\newblock A four-russian algorithm for regular expression pattern matching.
\newblock {\em J. ACM}, 39(2):430--448, 1992.

\bibitem{NR2003}
G.~Navarro and M.~Raffinot.
\newblock Fast and simple character classes and bounded gaps pattern matching,
  with applications to protein searching.
\newblock {\em J. Comp. Biology}, 10(6):903--923, 2003.

\bibitem{NH2011}
L.~Nielsen and F.~Henglein.
\newblock Bit-coded regular expression parsing.
\newblock In {\em Proc. 5th LATA}, pages 402--413, 2011.

\bibitem{SM2012}
M.~Sulzmann and K.~Z.~M. Lu.
\newblock Regular expression sub-matching using partial derivatives.
\newblock In {\em Proc. 14th PPDP}, pages 79--90, 2012.

\bibitem{Thomp1968}
K.~Thompson.
\newblock Regular expression search algorithm.
\newblock {\em Commun. ACM}, 11:419--422, 1968.

\bibitem{YCDLK2006}
F.~Yu, Z.~Chen, Y.~Diao, T.~V. Lakshman, and R.~H. Katz.
\newblock Fast and memory-efficient regular expression matching for deep packet
  inspection.
\newblock In {\em Proc. ANCS}, pages 93--102, 2006.

\end{thebibliography}
\appendix
\section{Proof of Recurrence}\label{sec:recurrenceproof} 

In this section we formally prove that the recurrence in Section~\ref{sec:analysis-rec} is $O(nm)$ for  $\gamma_n=2$ and $\gamma_m=25$. For simplicity, we have not attempted to minimize the constants. Recall the recurrence for $T(n,m)$: 
$$
T(n,m) = 
\begin{cases}
\sum_{i=1}^\ell T(n_i, \frac{2}{3}m+8)+ T(n_0 +\ell,\frac{2}{3}m+8) + O(mn) & m \geq \gamma_m \textrm{ and } n\geq \gamma_n  \\
 O(m+n) & m < \gamma_m  
 \textrm{ or }  n  <  \gamma_n 
\end{cases}
$$

We will show that  $T(n,m) \leq acmn - (a-1)cm +c$ for constants $a \geq 1$, $c \geq 1$ using induction. First consider the base case ($m < \gamma_m$ or $n<\gamma_n$). Since $mn + 1 \geq m+n$ and $T(n,m) = c(m+n)$ we have $T(n,m) \leq c(mn + 1) = cmn -cm + c(m+1)$ and it follows that $T(n,m) \leq acmn - (a-1)cm +c$. For the induction step ($m \geq \gamma_m$ and $n\geq \gamma_n$) we have 

\begin{eqnarray}
T(n,m) &=& \sum_{i=1}^\ell T(n_i, \frac{2}{3}m+8)+ T(n_0 +\ell,\frac{2}{3}m+8) + O(mn) \nonumber \\
&\leq& \sum_{i=1}^\ell \left(ac n_i( \frac{2}{3}m+8) - (a-1)c(\frac{2}{3}m+8)+c\right)  \nonumber \\
&& \qquad + \left(ac(n_0 +\ell)(\frac{2}{3}m+8)
 - (a-1)c(\frac{2}{3}m+8) + c\right) + cmn \nonumber \\
&=  & act\left(\frac{2}{3}m+8\right)\sum_{i=1}^\ell n_i  + ac(n_0 +\ell)(\frac{2}{3}m+8) \nonumber \\
&&\qquad  - (\ell+1)(a-1)c(\frac{2}{3}m+8)+(\ell+1) c  +cmn \nonumber \\
&\leq& act(\frac{2}{3}m+\frac{8}{25}m)\sum_{i=0}^\ell n_i + ac\ell (\frac{2}{3}m+\frac{8}{25}m)  \nonumber\\
&& \qquad - (\ell+1)(a-1)c(\frac{2}{3}m+\frac{8}{25}m)+(\ell+1) c  + cmn \label{eq:eq1} \\ %\quad //  m\geq 25 \\
&\leq& ac\frac{74}{75}mn+ac\ell\frac{74}{75}m - ac(\ell+1)\frac{74}{75}m  +(\ell+1)c\frac{74}{75}m+ (\ell+1) c + cmn \label{eq:eq2} \\ %//  m\geq 25 \\
%&=& ac\frac{74}{75}mn  - ac\frac{74}{75}m+ (\ell+1)c(\frac{74}{75}m+1)+cmn\\
%&\leq& ac\frac{74}{75}mn  - ac\frac{74}{75}m+ (\ell+1)c(\frac{74}{75}m+\frac{1}{25}m)+cmn\\
%&\leq& ac\frac{74}{75}mn  - ac\frac{74}{75}m+ (\ell+1)c\frac{77}{75}m+cmn\\
&\leq& ac\frac{74}{75}mn  - ac\frac{74}{75}m+ (n/2+2)c\frac{77}{75}m+cmn \label{eq:eq3} \\%//n\leq (n+1)/2\\ 
&\leq& acmn - acm +c(m+1) \label{eq:eq4}
\end{eqnarray}
Here, \eqref{eq:eq1} and \eqref{eq:eq2} follows since $m \geq 25$, \eqref{eq:eq3} follows since $\ell \leq (n+1)/2$, and \eqref{eq:eq4} holds for sufficiently large $a$.

\end{document}